\documentclass{article}
\usepackage{amsmath, amssymb, amsthm}
\usepackage{setspace}
\usepackage[a4paper, margin = 2.54cm]{geometry}
\usepackage{authblk}
\usepackage{graphicx}
\usepackage{fancyhdr}
\newtheorem{theorem}{Theorem}[subsubsection]

\newtheorem{definition}{Definition}

\theoremstyle{remark}

\numberwithin{equation}{section}
\numberwithin{definition}{section}
\numberwithin{theorem}{section}
\numberwithin{figure}{section}

\DeclareSymbolFont{bbold}{U}{bbold}{m}{n}
\DeclareSymbolFontAlphabet{\mathbbold}{bbold}
\begin{document}
\title{Absence of Binding in the Nelson and\\Piezoelectric Polaron Models}
\author{Gonzalo A. Bley}
\affil{Institut for Matematik, Aarhus Universitet, \authorcr Ny Munkegade 118,  8000 Aarhus C, Denmark}
\maketitle
\begin{abstract}
In the context of the massless Nelson model, we prove that two non-relativistic nucleons interacting with a massless meson field do not bind when a sufficiently strong Coulomb repulsion between the nucleons is added to the Hamiltonian. The result holds for both the renormalized and unrenormalized theories, and can also be applied to the so-called piezoelectric polaron model, which describes an electron interacting with the acoustical vibrational modes of a crystal through the piezoelectric interaction. The result can then be interpreted as well as a no-binding statement about piezoelectric bipolarons. The methods used allow also for a significant reduction of about 30\% over previously known no-binding conditions for the optical bipolaron model of H. Fr\"{o}hlich.
\end{abstract}
\begin{section}{Introduction}
The intent of this article is to provide a simple proof that two particles attracting each other via an effective massless Nelson-model interaction do not bind when they repel one another through a strong enough Coulomb repulsion. Even though in the Nelson model the particles involved are nucleons (and so need not repel each other), there is a model of two electrons interacting with the acoustical vibrational modes of a piezoelectric crystal, the so-called piezoelectric bipolaron, that has exactly the same Hamiltonian as the massless Nelson model for two particles, with the addition of a Coulomb repulsion term. For this reason our result is of direct physical significance. In the following subsection we shall describe the models involved, while providing references for them.
\begin{subsection}{The Models Involved}
\label{subsection.models.involved}
\begin{subsubsection}{The Massless Nelson Model}
The massless Nelson model describes the interaction of a system of $N$ non-relativistic nucleons with a quantized, massless meson field. Its Hamiltonian reads as
\begin{gather}
H = -\sum_{n = 1}^N\frac{\Delta_n}{2} + \int_{\mathbb{R}^3}\chi_{\Lambda}(k)|k|a_k^{\dagger}a_k\,dk + \sqrt{\alpha}\sum_{n = 1}^N\int_{\mathbb{R}^3}\frac{\chi_{\Lambda}(k)}{\sqrt{|k|}}\left(e^{ikx_n}a_k + e^{-ikx_n}a_k^{\dagger}\right)\,dk,
\label{equation.hamiltonian.nelson.model}
\end{gather}
acting on $L^2(\mathbb{R}^{3N})\otimes\mathcal{F}$, where $\mathcal{F}$ is the Fock space on $L^2\left(\mathbb{R}^3\right)$. $a_k$ and $a_k^{\dagger}$ denote the annihiliation and creation operators for the $k$-th mode of the field, respectively, and satisfy the relations $[a_k, a_{k'}^{\dagger}] = \delta\left(k - k'\right)$ and $\left[a_k, a_{k'}\right] = [a_k^{\dagger}, a_{k'}^{\dagger}] = 0$. $\alpha$ is a non-negative number, the coupling constant of the interaction between the nucleons and the meson field; $\Lambda$ is also a non-negative number, and it acts as an ultraviolet cutoff; and $\chi_{\Lambda}$ is the indicator function of the ball of radius $\Lambda$ centered at the origin. $\Delta_n$ is simply the part of the full $3N$-Laplacian corresponding to the particle $n$, namely if $X = (x_1, x_2, \ldots, x_N)$ is the position vector of the particles (a $3N$-dimensional vector), and if $x_n = (x_n^1, x_n^2, x_n^3)$, then $\Delta_n = \partial^2/\partial^2 x_n^1 + \partial^2/\partial^2 x_n^2 + \partial^2/\partial^2 x_n^3$. The model is due to Edward Nelson, who presented it for the first time in 1964 in \cite{N1}, and subsequently in \cite{N2}. The most noteworthy property of the Hamiltonian \eqref{equation.hamiltonian.nelson.model} is that it requires renormalization if one is to make sense of it when $\Lambda = \infty$: even though $H$ is self-adjoint and bounded-below when $\Lambda$ is finite, the ground-state energy of $H$ goes to negative infinity as $\Lambda \to \infty$. However, if one adds a term, that we denote by $Q$ and define as
\begin{gather}
\alpha N\int_{\mathbb{R}^3}\frac{\chi_{\Lambda}(k)}{|k|^2\left(|k|/2 + 1\right)}\,dk = 8\pi\alpha N\log\left(1 + \Lambda/2\right),
\label{equation.renormalizing.term}
\end{gather}
then one obtains a self-adjoint, bounded below operator $\widehat{H}$ in the limit; more precisely, for every real $t$, $e^{it(H_{\Lambda} + Q_{\Lambda})}$ goes to $e^{it\widehat{H}}$ strongly as $\Lambda$ goes to infinity. The aforementioned properties of \eqref{equation.hamiltonian.nelson.model} were observed by Nelson \cite{N2}. We would like to mention as well that an additional infrared cutoff $\left\{ k \in \mathbb{R}^3 : |k| \geq \mu\right\}$, on top of $\chi_{\Lambda}(k)$, is required for $H$ to have a ground state \cite{LMS}, even though this fact will not be relevant for us in the present article. We are interested here in the case $N = 2$.
\end{subsubsection}
\begin{subsubsection}{The Piezoelectric Polaron}
The piezoelectric polaron describes the interaction of an electron with the acoustical vibrational modes of a piezoelectric crystal. It is usually attributed to R.A. Hutson \cite{H}, and to G.D. Mahan and J.J. Hopfield \cite{MH}. Its Hamiltonian is the same as that of the massless Nelson model,
\begin{gather}
H^1 \equiv -\frac{\Delta}{2} + \int_{\mathbb{R}^3}\chi_{\Lambda}(k)|k|a_k^{\dagger}a_k\,dk + \sqrt{\alpha}\int_{\mathbb{R}^3}\frac{\chi_{\Lambda}(k)}{\sqrt{|k|}}\left(e^{ikx}a_k + e^{-ikx}a_k^{\dagger}\right)\,dk,
\label{equation.piezoelectric.polaron}
\end{gather}
where now $\Lambda$ is kept fixed at a positive value, the so-called Debye wave number \cite[Page 430, Footnote 6]{L}. $\alpha$ is defined in terms of quantities that depend on the crystal in question (such as the speed of sound); see \cite{TW, WGT} for a precise definition of this constant, and also \cite{RW, WP} to gain a better understanding of the model. We will study here the case of two electrons in a piezoelectric crystal, whose Hamiltonian follows directly from \eqref{equation.piezoelectric.polaron},
\begin{gather}
H_A^2 \equiv -\sum_{n = 1}^2\frac{\Delta_n}{2} + \int_{\mathbb{R}^3}\chi_{\Lambda}(k)|k|a_k^{\dagger}a_k\,dk + \sqrt{\alpha}\sum_{n = 1}^2\int_{\mathbb{R}^3}\frac{\chi_{\Lambda}(k)}{\sqrt{|k|}}\left(e^{ikx_n}a_k + e^{-ikx_n}a_k^{\dagger}\right)\,dk + \frac{A}{|x_1 - x_2|}.
\label{equation.piezoelectric.bipolaron}
\end{gather}
Physically $A$ is, after fully restoring units, $e^2/(4\pi\varepsilon)$, where $e$ is the electron charge and $\varepsilon$ is the permittivity of the medium the electrons are in (not the vacuum, but a crystal), but it will be more transparent for us to simply treat $A$ as a fixed constant that can take any non-negative value. Here one may treat the two electrons as fermions or simply not impose any symmetry on them -- our final result will be valid in both cases.
\end{subsubsection}
\begin{subsubsection}{The Optical Polaron}
The last model we will have the opportunity to discuss is that of the optical polaron of H. Fr\"{o}hlich \cite{Fr2, Fr}. It describes the interaction of a single non-relativistic electron with the optical vibrational modes of a crystal lattice. Its Hamiltonian is similar to the ones just described, the main difference being that there is no ultraviolet cutoff,
\begin{gather}
-\frac{\Delta}{2} + \int_{\mathbb{R}^3}a_k^{\dagger}a_k\,dk + \frac{\sqrt{\alpha}}{2^{3/4}\pi}\int_{\mathbb{R}^3}\frac{1}{|k|}\left(a_k e^{ikx} + a_k^{\dagger}e^{-ikx}\right)\,dk.
\label{equation.polaron.hamiltonian}
\end{gather}
Even though the model was first devised by H. Fr\"{o}hlich, many people after him provided critical contributions to its understanding. One of them was R. Feynman, who in 1955 \cite{F} provided a new interpretation of the interaction of the electron with the lattice through the use of functional integrals. In particular, the functional integral analysis reveals a fact that is not visible at the level of the operator \eqref{equation.polaron.hamiltonian}, which is that the electron, roughly speaking, is attracted to its own past history via a Coulomb interaction with coupling constant $\alpha$. Relevant here as well are works by T.D. Lee, F.E. Low, D. Pines \cite{LLP, LP}, and M. Gurari \cite{G}, where a variational principle due to S. Tomonaga \cite{T} was used to obtain a power series expansion of the ground-state energy of the polaron in terms of the total momentum of the system. (In reality they obtained an upper bound, since they used a trial state in their analysis.) We would like also to mention perturbation-theoretic calculations for the polaron due to E. Haga \cite{Ha}, and a work by E.H. Lieb and K. Yamazaki, where a rigorous lower bound to the polaron energy was found \cite{LY}. These are all early works from the 50's. Important for us in the present article is a more recent work from R.L. Frank, E.H. Lieb, R. Seiringer, and L.E. Thomas \cite{FLST} where, in particular, a no-binding condition for the optical bipolaron was found. We will refer to this last work repeatedly during the rest of the paper. It provides, in particular, good references on more recent works on the model. The Hamiltonian \eqref{equation.polaron.hamiltonian} is to be interpreted as the norm-resolvent limit of the corresponding operator with an ultraviolet cutoff in the interaction \cite[Discussion immediately above (1.2)]{GM}, and is the form used in the aforementioned work of Frank et al. \cite{FLST}.
\end{subsubsection}
\end{subsection}
\begin{subsection}{Methods of Proof}
In the following discussion, we shall refer to the piezoelectric polaron model. A simple argument involving functional integrals, which we will elaborate on later, shows that the interaction of the two electrons with each other when the repulsion parameter $A$ is equal to $0$, mediated through their coupling with the acoustic phonon field, is attractive and ``retarded Coulomb-like'' (the precise meaning of which we shall see later on) with strength $\alpha$ -- up to constants. If $A$ is big enough (or, equivalently, if $\alpha$ is small enough), one would expect that in the minimum energy configuration the electrons would be pushed very far apart by the mutual Coulomb repulsion (which would overcome, by virtue of its strength, the attraction created by the interaction of the particles with the field), and that the electrons would be left interacting by themselves with their own local cloud of excitations of the phonon field. In particular, one would expect that there would be no binding between the two polarons, meaning that the ground-state energy of the entire system of two electrons immersed in the crystal would be equal to twice the ground-state energy of the Hamiltonian above for one piezoelectric polaron, Equation \eqref{equation.piezoelectric.polaron}.

Exactly how much bigger $A$ would have to be in relation to $\alpha$ for there not to be binding would depend on the intricate and exact nature of the attraction between the two electrons; and since this is by no means a simple attraction, we were merely able to prove that if $A \geq C(\Lambda)\alpha$, where $C(\Lambda)$ is an explicit but diverging function of $\Lambda$, then the ground-state energy of the full Hamiltonian \eqref{equation.piezoelectric.bipolaron} is twice the ground-state energy of the corresponding single-particle Hamiltonian \eqref{equation.piezoelectric.polaron}. We certainly do not endeavor into making any statements as to the sharpness of this relation, $A \geq C(\Lambda)\alpha$, since we simply do not know what happens if $A$ is smaller than $C(\Lambda)\alpha$. Even though our function $C(\Lambda)$ goes to $\infty$ as $\Lambda \to \infty$, this is not a problem for the piezoelectric polaron, since one keeps $\Lambda$ finite, but it is an issue for the Nelson model, as in that case one does take $\Lambda$ to $\infty$ in order to renormalize it. We relegate a no-binding result for the Nelson model to the last section of this article, and the reason why we leave it for last is that we do not get what one would like, or expect -- we obtain a non-linear relationship between $A$ and $\alpha$. This is most likely an artifact; see Section \ref{section.polaron.nelson}.

The method of proof of the results above involves three ingredients: First, a partition of unity of the configuration space $\mathbb{R}^6$ of the position of the two electrons, which allows local estimates to be made. The partition is an adaptation to the piezolectric bipolaron of an argument of Frank, Lieb, Seiringer, and Thomas \cite{FLST}. Second, a refined lower bound for the ground-state energy of two particles interacting with each other and with a quantum field through the Nelson interaction (without repelling each other), given in a recent paper of the author \cite{B}, as well as an upper bound for the corresponding 1-particle model, provided in the present article; and third, the observation that the massless Nelson interaction is essentially Coulomb when the interparticle distance is localized, which becomes apparent at the functional integral level -- see Section \ref{section.second.localization}.

Regarding the first ingredient, a double partition of unity was performed in \cite{FLST}, where the interparticle distance was first split using a single length scale that is then raised to ever higher powers for large distances, in order to control the localization error; a second partition for the entire space $\mathbb{R}^6$ was then made using two movable balls of fixed radius, in order to localize the electrons even further in their own boxes. In total, they localize, so to speak, 7 degrees of freedom (the interparticle distance, which is a one-dimensional object, plus the center of each ball, gives a total of 7), which is obviously an ``overlocalization'' (there being only 6 spatial degrees of freedom), and one would expect to be able to solve the problem without localizing so much. Indeed, we show how one can make do with just localizing one of the balls, and not the two of them. Thus, in total, we localize just 4 degrees of freedom. All the partitions are then completely optimized using a reduction of the resulting infinite dimensional minimization problem to a low-dimensional one, as in \cite{BB}.

When our methods are applied to the optical bipolaron model of H. Fr\"{o}hlich -- the obvious extension of \eqref{equation.polaron.hamiltonian} to 2 electrons -- we obtain a significant improvement over previous results on the no-binding of electrons in this model, namely the passage from $A \geq 36.9\alpha$ \cite{BB} to $A \geq 25.9\alpha$. (In \cite{FLST} the condition was $A \geq 37.7\alpha$.) See Section \ref{section.polaron.nelson}. (In \cite{BB} and \cite{FLST} the no-binding conditions appear with different numbers because $p^2$ was used for the kinetic energy instead of $p^2/2$.)

We would like to finish this subsection by pointing out a result known as the subadditivity of the energy. For our purposes here, it merely states that $\text{inf spec }H_A^2 \leq 2 \text{ inf spec }H^1$, where $H^1$ and $H_A^2$ have been defined in the previous subsection. This result holds even if one treats the electrons as fermions (meaning that the infimum on the left is on anti-symmetric functions). The proof is a careful execution of what we explained in the first paragraph of this subsection; the idea is to separate the two electrons as much as possible, so that they essentially interact just with their own phonon cloud. See \cite[Section 14.2]{LS} and the references in \cite[Theorem 1.4 and its proof]{GM} for more information. What is missing to prove absence of binding is then the other inequality $\text{inf spec }H_A^2 \geq 2 \text{ inf spec }H^1$, and that is what this article provides. One then gets that the binding energy, $2\text{ inf spec }H^1 - \text{inf spec } H_A^2$, is zero. The same arguments hold for the optical polaron model (and the Nelson model).
\end{subsection}
\begin{subsection}{Main Results}
Even though the results in the article have been already hinted at in the previous subsection, we shall state them now precisely, for the convenience of the reader. We shall start by formally defining a concept we have been alluding to repeatedly.
\begin{definition}[Absence of binding]
Let $E_1$ be the ground-state energy of any of the 3 models mentioned in Subsection \ref{subsection.models.involved} when only one particle is present, and $E_2(A)$ be the corresponding ground-state energy when there are 2 particles and an additional Coulomb-repulsion term $A/|x_2 - x_1|$ is considered. We say that there is no binding if the binding energy $2E_1 - E_2(A)$ of the two-particle system is zero.
\end{definition}
Intuitively, this concept means that the minimum energy configuration of the two-particle system is that of two particles infinitely separated from one another, interacting with their own local cloud of excitations of the quantum field.
\begin{theorem}[Piezoelectric polaron model]
In the case of the piezoelectric bipolaron model, Hamiltonian \eqref{equation.piezoelectric.bipolaron}, there is an explicit positive function of the cutoff $\Lambda$, $C$, such that there is no binding if $A \geq C(\Lambda)\alpha$.
\label{theorem.main.result.piezoelectric.bipolaron}
\end{theorem}
\begin{proof}
See Sections \ref{section.partition} and \ref{section.second.localization}. The function $C$ appears in Equation \eqref{equation.condition.for.no.binding}.
\end{proof}
\begin{theorem}[Nelson model]
There is no binding in the 2-body massless Nelson model, Equation \eqref{equation.hamiltonian.nelson.model} with $N = 2$, if a repulsion term $A/|x_1 - x_2|$ is added to the Hamiltonian, with $A \geq B_1\alpha + B_2\alpha^7$, for some explicit positive constants $B_1$ and $B_2$, which are independent of the cutoff $\Lambda$. The result holds for both the unrenormalized and renormalized theories.
\end{theorem}
\begin{proof}
See Section \ref{section.polaron.nelson}. The constants $B_1$ and $B_2$ can be derived from \eqref{equation.condition.no.binding.nelson}. The explanation of the meaning that the result is true for both the renormalized and unrenormalized theories is found in the paragraph containing Equation \eqref{equation.energy.renormalization.nelson.two.particles}.
\end{proof}
\begin{theorem}[Polaron model]
For the optical bipolaron model of H. Fr\"{o}hlich (see Equation \eqref{equation.polaron.hamiltonian} and consider two particles, in the spirit of Hamiltonian \eqref{equation.piezoelectric.bipolaron}), one has no binding as soon as $A \geq 25.9\alpha$.
\end{theorem}
\begin{proof}
See Section \ref{section.polaron.nelson}.
\end{proof}
As already mentioned, the bound $A \geq 25.9\alpha$ is an improvement of about $30\%$ over previous results on the no-binding of bipolarons \cite{FLST, BB}.
\end{subsection}
\begin{subsection}{Remark on Some Functional Integrals}
\label{subsection.remark.functional.integrals}
Our proofs below will rely heavily on the use of functional integrals for the estimation of ground-state energies. In particular, for the two-electron piezoelectric polaron (or two-nucleon massless Nelson model with repulsion) we have that the ground-state energy is bounded from below by
\begin{align}
-\limsup_{T \to \infty}\frac{1}{T}&\log\left\{\sup_{(x, y) \in \mathbb{R}^6}E^{(x, y)}\left[\exp\left(\alpha\sum_{m, n = 1}^2\int\!\!\!\int_0^T\!\!\!\int_0^t\chi_{\Lambda}(k)e^{-|k|(t - s)}e^{-ik(X_t^m - X_s^n)}|k|^{-1}\,ds\,dt\,dk\right.\right.\right.\nonumber\\
& \left.\left.\left. \qquad\qquad\qquad\qquad - A\int_0^T\frac{dt}{|X_t^1 - X_t^2|}\right)\right]\right\},
\label{equation.feynman.kac.piezoelectric.polaron}
\end{align}
where $X = (X^1, X^2)$ is $6 = 3 + 3$-dimensional Brownian motion starting at $(x, y)$, and $E^{(x, y)}$ denotes expectation with respect to that process. As for the optical bipolaron model, the corresponding lower bound is
\begin{gather}
-\limsup_{T \to \infty}\frac{1}{T}\log\left\{\sup_{(x, y) \in \mathbb{R}^6}E^{(x, y)}\left[\exp\left(\frac{\alpha}{\sqrt{2}}\sum_{m, n = 1}^2\int_0^T\!\!\!\int_0^t\frac{e^{-(t - s)}}{|X_t^m - X_s^n|}\,ds\,dt - A\int_0^T\frac{dt}{|X_t^1 - X_t^2|}\right)\right]\right\}.
\label{equation.feynman.kac.polaron}
\end{gather}
These two estimates follow basically from the analysis contained in \cite[Appendix A]{B} and \cite[Chapter 2]{B2}. Noteworthy is the fact that the quantum field variables have disappeared in the two Feynman-Kac-like formulas above. The expectation in \eqref{equation.feynman.kac.piezoelectric.polaron} was basically known to Nelson in his first work on his model \cite{N1} -- a functional integral analysis of the model was in fact his first approach to the Hamiltonian \eqref{equation.hamiltonian.nelson.model}, that he left behind in favor of operator methods \cite{N2}. The expectation in \eqref{equation.feynman.kac.polaron} was found for the first time by Feynman \cite{F} in the case of a single electron, by integrating the quantum field variables, using methods developed in \cite{F2}. These two estimates, and variations of them, will be used throughout the rest of article.

We shall make use also of the following exact Feynman-Kac formulas for the ground-state energies of the two-electron piezolectric polaron and optical bipolaron models, respectively,
\begin{gather}
-\lim_{R \to \infty}\lim_{T \to \infty}T^{-1}\log\left[\int_{B_R}\int_{B_R}\int\!\int\exp\left(\alpha\sum_{m, n = 1}^2\int\!\!\!\int_0^T\!\!\!\int_0^t\chi_{\Lambda}(k)e^{-|k|(t - s)}e^{-ik(\omega_t^m - \omega_s^n)}|k|^{-1}\,ds\,dt\,dk\right.\right.\nonumber\\
\left.\left.\qquad\qquad\qquad\qquad - A\int_0^T\frac{dt}{|\omega_t^1 - \omega_t^2|}\right)\eta_R(\omega^1)\eta_R(\omega^2)\,dW_{x, x}^T(\omega^1)\,dW_{y, y}^T(\omega^2)\,dx\,dy\right],\label{equation.feynman.kac.exact.piezoelectric.polaron}\\
-\lim_{R \to \infty}\lim_{T \to \infty}T^{-1}\log\left[\int_{B_R}\int_{B_R}\int\!\int\exp\left(\frac{\alpha}{\sqrt{2}}\sum_{m, n = 1}^2\int_0^T\!\!\!\int_0^t\frac{e^{-(t - s)}}{|\omega_t^m - \omega_s^n|}\,ds\,dt\right.\right.\nonumber\\
\left.\left. \qquad\qquad\qquad\qquad - A\int_0^T\frac{dt}{|\omega_t^1 - \omega_t^2|}\right)\eta_R(\omega^1)\eta_R(\omega^2)\,dW_{x, x}^T(\omega^1)\,dW_{y, y}^T(\omega^2)\,dx\,dy\right],\label{equation.feynman.kac.exact.polaron}
\end{gather}
where $B_R$ is the ball centered at the origin of radius $R$, $\omega^1$ and $\omega^2$ are independent 3D Brownian motions, $\eta_R$ is the indicator function equal to 1 if a Brownian path $\omega$ is completely contained in $B_R$ and 0 otherwise, and $dW_{x, x}^T$ is conditional Wiener measure for Brownian paths that start and end at the point $x$ in $\mathbb{R}^3$.

Some words pertaining Formulas \eqref{equation.feynman.kac.exact.piezoelectric.polaron} and \eqref{equation.feynman.kac.exact.polaron} are in order now. They can be obtained from the analysis found in a book by G. Roepstorff \cite{R}, regarding the computation of the partition function for a system consisting of a particle linearly coupled to a Bose field; specifically, the one found in Sections 5.1 and 5.3. The formulas follow basically from selecting $V(x)$ equal to $\infty$ if $x$ is not in the ball $B_R$, and equal to 0 otherwise, in \cite[Equation (5.3.19)]{R}. The expressions inside the exponentials in \eqref{equation.feynman.kac.exact.piezoelectric.polaron} and \eqref{equation.feynman.kac.exact.polaron} are identical to the ones in \eqref{equation.feynman.kac.piezoelectric.polaron} and \eqref{equation.feynman.kac.polaron}, because the way the field variables are integrated in \cite[Section 5.1]{R} is equivalent to the one in \cite[Section 2.1]{B2}. Equation \eqref{equation.feynman.kac.exact.polaron} for the polaron was used in the work by Frank, Lieb, Seringer, and Thomas alluded to before \cite[Equations (1.22) and (1.23)]{FLST}.
\end{subsection}
\begin{subsection}{The Structure of the Article and Acknowledgments}
We now give an outline of the article. In Section \ref{section.partition} we partition the distance between electrons in the context of the piezoelectric bipolaron. This is the first localization. In Section \ref{section.second.localization} we continue referring to the piezoelectric bipolaron, and another localization is performed, where a single electron is placed in a ball, thus ``pinning'' it to a center. This second localization allows the two electrons to stay far apart, even with the Coulomb-like attraction between them being present, which arises from the coupling with the field. (This is not a totally trivial fact, as will become clear later in the paper.) The final result $A \geq C(\Lambda)\alpha$ is given in this section. Then, in Section \ref{section.polaron.nelson} we study what happens when the method used for the piezoelectric polaron is mimicked in the optical bipolaron and Nelson models. In particular, we obtain an improvement over previous results on no-binding of optical bipolarons, bringing the condition $A \geq 36.9\alpha$ to $A \geq 25.9\alpha$, as we have previously mentioned. In Appendix A we provide a short description of the techniques behind the lower bounds for the spectra of the models involved in this work. In Appendix B we provide a short proof of an upper and a lower bound for the massless Nelson model that are used in the no-binding proof for the piezoelectric polaron. In Appendix C we explain and address a few mistakes made in the Ph.D. thesis of the author, on which the present article is partially based.

We would like to take this last paragraph to thank Lawrence Thomas for very long and productive discussions. The author acknowledges as well partial support from the Danish Council for Independent Research (Grant number DFF-4181-00221).
\end{subsection}
\end{section}
\begin{section}{Partition of Interparticle Distance.}
\label{section.partition}
We will focus in this and the following section on the piezoelectric polaron model. Only in Section \ref{section.polaron.nelson} will we refer to the optical polaron and Nelson models. In this section we perform a partition of unity on the configuration space $\mathbb{R}^6$ for the position $(x, y)$ of two three-dimensional particles. The construction here follows the lines in \cite[Section 2]{FLST}, adapted to our purposes for the piezoelectric polaron. Let $a_0, a_1, a_2, \ldots$ be positive numbers and define, for each $n \geq 0$, $s_n \equiv \sum_{i = 0}^n a_i$. We partition the half real-line $[0, \infty)$ with the functions
\begin{align}
\varphi_0(t) & \equiv
\begin{cases}
1 & 0 \leq t \leq a_0\\
\displaystyle\cos\left[\frac{\pi(t - a_0)}{2a_1}\right]\qquad  & a_0 \leq t \leq a_0 + a_1,
\end{cases}
\end{align}
and
\begin{align}
\varphi_n(t) & \equiv
\begin{cases}
\displaystyle\sin\left[\frac{\pi\left(t - s_{n - 1}\right)}{2a_n}\right] \qquad & s_{n - 1} \leq t \leq s_n\vspace{1mm}\\
\displaystyle\cos\left[\frac{\pi\left(t - s_n\right)}{2a_{n +1}}\right] & s_n \leq t \leq s_{n + 1}.
\end{cases}
\end{align}
(Each $\varphi_n$ is defined as 0 outside of the intervals given above.) Each one of the functions $\varphi_n$ represents an asymmetric bump in the shape of a sine function with one side longer than the other, with the exception of $\varphi_0$, which is the shape of a half-pill. By construction, $\sum_{n = 0}^{\infty}\varphi_n^2(t)$ is equal to 1 for all $t$, and so the functions $\varphi_n$ form a quadratic partition of unity for $[0, \infty)$. We then use the functions $\varphi_n$ to separate the interparticle distance, resulting in a partition of $\mathbb{R}^6$; namely, we consider the functions $\phi_n(x, y) \equiv \varphi_n(|x - y|)$ for $n \geq 0$. And since obviously $\sum_{n = 0}^{\infty}\phi_n^2 = 1$, the collection of functions $\phi_0, \phi_1, \phi_2, \ldots$ forms a quadratic partition of unity for $\mathbb{R}^6$. We then have, by the IMS formula \cite[Section 3.1]{CFKS},
\begin{gather}
H_A^2 = \sum_{n = 0}^{\infty}\phi_n H_A^2\phi_n - \frac{1}{2}\sum_{n = 0}^{\infty}|\nabla\phi_n|^2,
\label{equation.IMS}
\end{gather}
where $H_A^2$ is the Hamiltonian of the piezoelectric bipolaron, Equation \eqref{equation.piezoelectric.bipolaron}. The second term in \eqref{equation.IMS} is a localization error. It tells us that localizing comes at an increase in kinetic energy (which is expected, given the uncertainty principle).

For $\psi$ as any state in the quadratic form domain of $H_A^2$ one has, by defining $\psi_n \equiv \phi_n\psi$,
\begin{align}
(\psi, H_A^2\psi) & = \sum_{n = 0}^{\infty}(\psi_n, H_A^2\psi_n) - \frac{1}{2}\sum_{n = 0}^{\infty}\left(\psi, |\nabla\phi_n|^2\psi\right)\nonumber\\
& = \sum_{n = 0}^{\infty}\left(\psi_n, H_0^2\psi_n\right) - \frac{1}{2}\sum_{n = 0}^{\infty}\left(\psi, |\nabla\phi_n|^2\psi\right) + A\sum_{n = 0}^{\infty}\left(\psi_n, |x - y|^{-1}\psi_n\right)\nonumber\\
& \geq (\psi_0, H_0^2\psi_0) + \sum_{n = 1}^{\infty}(\psi_n, H_0^2\psi_n) - \frac{1}{2}\sum_{n = 0}^{\infty}(\psi, |\nabla\phi_n|^2\psi) + A\sum_{n = 0}^{\infty}\frac{\|\psi_n\|^2}{s_{n + 1}}.
\label{equation.first.localized.bound}
\end{align}
The idea now is to prove that, if $A$ is big enough,
\begin{equation}
(\psi, H_A^2\psi) \geq 2\text{ inf spec }H^1\sum_{n = 0}^{\infty}\|\psi_n\|^2 = 2\text{ inf spec }H^1\|\psi\|^2,
\label{equation.no.binding.condition}
\end{equation}
which will imply no-binding. (Recall that $H^1$ is the analog of Equation \eqref{equation.piezoelectric.bipolaron} for just one particle, Equation \eqref{equation.piezoelectric.polaron}.) We will accomplish this by bounding each one of the terms in \eqref{equation.first.localized.bound}. We will start with the third one -- it can be controlled as follows: by noticing that $|\nabla\phi_n(x, y)|^2 = 2|\varphi_n'(|x - y|)|^2$ for all $n$ and recalling that $\sum_{m = 0}^{\infty}\phi_m^2 = 1$,
\begin{align}
\sum_{n = 0}^{\infty}(\psi, |\nabla\phi_n|^2\psi) & = 2\sum_{n = 0}^{\infty}(\psi, |\varphi_n'(|x - y|)|^2\psi) = 2\sum_{m = 0}^{\infty}\sum_{n = 0}^{\infty}\left(\psi_m, |\varphi_n'(|x - y|)|^2\psi_m\right)\nonumber\\
& \leq \frac{\pi^2}{a_1^2}\|\psi_0\|^2 + \sum_{m = 1}^{\infty}\frac{\pi^2}{\min(a_{m + 1}, a_m)^2}\|\psi_m\|^2.
\end{align}
We now continue with the first term in \eqref{equation.first.localized.bound}. This corresponds to the case where the electrons are close to each other. We provide in Appendix B the following lower and upper bounds for the piezoelectric polaron: The ground-state energy of the two-particle piezoelectric polaron without repulsion is bounded below by $-2C_1(\Lambda)\alpha - 8C_2(\Lambda)\alpha^2$, where $C_1$ and $C_2$ are explicit but diverging positive functions of $\Lambda$, and the one-particle piezoelectric polaron ground-state energy, which we denote by $E$ (equal to $\text{inf spec }H^1$, using the notation of \eqref{equation.no.binding.condition}), is bounded above by $-C_1(\Lambda)\alpha$, where $C_1$ is the same as in the lower bound. Then
\begin{gather}
\left(\psi_0, H_0^2\psi_0\right) \geq \text{inf spec }H_0^2\|\psi_0\|^2 \geq \left[-2C_1(\Lambda)\alpha - 8C_2(\Lambda)\alpha^2\right]\|\psi_0\|^2 \geq \left[2E - 8C_2(\Lambda)\alpha^2\right]\|\psi_0\|^2.
\label{equation.estimate.first.region}
\end{gather}
By grouping terms, we summarize what has been done so far -- from Equation \eqref{equation.first.localized.bound},
\begin{align}
\left(\psi, H_A\psi\right) \geq & \left[2E - 8C_2(\Lambda)\alpha^2 - \frac{\pi^2}{2a_1^2} + \frac{A}{a_0 + a_1}\right]\|\psi_0\|^2\nonumber\\
& + \sum_{n = 1}^{\infty}\left(\psi_n, H_0^2\psi_n\right) + \sum_{n = 1}^{\infty}\left[\frac{A}{s_{n + 1}} - \frac{\pi^2}{2\min(a_{n + 1}, a_n)^2}\right]\|\psi_n\|^2.
\end{align}
This concludes the first part of the bounding of the terms in \eqref{equation.first.localized.bound}. What remains is the bounding of the energy expectations where the electrons are far apart, the terms $(\psi_n, H_0\psi_n)$ for $n \geq 1$. An additional localization will be performed to control these expectations in the next section.
\end{section}
\begin{section}{Further Localization: Single-Electron Pinning}
\label{section.second.localization}
Let $n \geq 1$. We will spend this section bounding from below the term $(\psi_n, H_0^2\psi_n)$. We will perform a second localization where one of the electrons will be effectively pinned down, which will allow us at the end to arrive at a lower bound. This localization will be made to only one of the electrons, but either of them may be selected -- we will pick the ``second one'' (the one with $y$-coordinates). In \cite[Section 2]{FLST} the two electrons were pinned in a symmetrical fashion, which introduced an extra localization error with respect to what we do here.

For the space $\mathbb{R}^3$ we construct the following partition: Let $f$ be a function $\mathbb{R}^3 \to \mathbb{R}$ with the following properties: $f$ is continuous, $f$ has compact support, $f$ is $C^1$ on its support, and $\|f\|_2 = 1$. Then consider the family of functions $f_u(x, y) \equiv f(y - u)$. Then obviously $\int_{\mathbb{R}^3}f_u^2\,du = 1$, and so the family forms a continuous partition of unity. We then have the following formula,
\begin{gather}
H_0^2 = \int_{\mathbb{R}^3}f_u H_0^2 f_u\,du - \frac{1}{2}\int_{\mathbb{R}^3}|\nabla f_u |^2\,du.
\label{equation.second.localization}
\end{gather}
(This follows from a proof analogous to the one found in \cite[Section 3.1]{CFKS}.) The symmetry of the problem at hand will make it at the end very natural to select a sphere of a certain radius, say $R_n$, as the support of $f$, and so it will be fixed at that. One would like now to make the localization error, the second term on the right side of the equation above, as small as possible. This is just equal to $\frac{1}{2}\int_{\mathbb{R}^3}|\nabla f(y)|^2\,dy$, the infimum of which over all functions $f$ with the aforementioned properties, and support equal to a sphere of radius $R_n$ centered at the origin, is just the infimum of the spectrum of $-\Delta/2$ with Dirichlet conditions on the boundary of the sphere. We then choose $f$ to be equal to the ground-state of the aforementioned operator on the sphere of radius $R_n$,
\begin{gather}
f(y) \equiv
\begin{cases}
\displaystyle\frac{\sin\left(\pi |y|/R_n\right)}{\sqrt{2\pi R_n}|y|}\qquad & |y| \leq R_n\\
0 \qquad & |y| > R_n.
\end{cases}
\end{gather}
Since $-\Delta f = \left(\pi^2/R_n^2\right) f$, we conclude that $\frac{1}{2}\int_{\mathbb{R}^3}|\nabla f_u|^2\,du$ is bounded below by $\pi^2/(2R_n^2)$ and that it can be made equal to this value. We furthermore pick $R_n$ to satisfy the relationship $2R_n < s_{n - 1}$. From Equation \eqref{equation.second.localization}, if $\psi_{n, u}$ denotes the function $\psi\phi_n f_u$, we then obtain the formula
\begin{gather}
(\psi_n, H_0^2\psi_n) = \int_{\mathbb{R}^3}(\psi_{n, u}, H_0^2\psi_{n, u})\,du - \frac{\pi^2}{2R_n^2}\|\psi_n\|^2.
\label{equation.second.localization.performed}
\end{gather}

With the second localization, we have accomplished the following: the ``second'' electron has been pinned down to a ball in 3-space, and the ``first'' electron lies away from the second one in a shell that encloses the ball just mentioned, staying always at a certain distance from it, without intersecting it. The set where the electrons are may be described as $\Upsilon_{n, u} \equiv \left\{(x, y) \in \mathbb{R}^6 : s_{n - 1} \leq |x - y| \leq s_{n + 1}, |y - u| \leq R_n\right\}$, and what has been obtained with this is that the electrons have been effectively separated, as the following inclusion shows,
\begin{gather}
\Upsilon_{n, u} \subset \left\{x \in \mathbb{R}^3 : s_{n - 1} - R_n \leq |x - u| \leq s_{n + 1} + R_n\right\}\times\left\{y \in \mathbb{R}^3 : |y - u| \leq R_n\right\}.
\end{gather}
Note that, since $2R_n < s_{n - 1}$, these two last subsets of $\mathbb{R}^3$ do not intersect, and so the cartesian product above can be embedded in $\mathbb{R}^3$ as the union of the two. If we now let $V_{n, u}$ be the separating potential corresponding to $\Upsilon_{n, u}$, namely $V_{n, u}(x, y) = 0$ if $(x, y) \in \Upsilon_{n, u}$ and $V_{n, u} = \infty$ otherwise, by noting that $V_{n, u}$ commutes with the potentials of $H_0$, we get the following Feynman-Kac formula,
\begin{align}
& \, (\psi_{n, u}, H_0\psi_{n, u})\nonumber\\
= & \left(\psi_{n, u}, \left(H_0 + V_{n, u}\right)\psi_{n, u}\right)\nonumber\\
\geq & \,\,\text{inf spec }(H_0 + V_{n, u})\|\psi_{n, u}\|^2\nonumber\\
= & -\lim_{R \to \infty}\lim_{T \to \infty}T^{-1}\log\left[\int_{B_R}\!\int_{B_R}\int\!\int\exp\left(\sum_{i, j = 1}^2\mathcal{A}_{i, j}^T\right)\right.\nonumber\\
&\qquad\qquad\qquad\qquad\qquad\times\left.\exp\left(-\int_0^T V_{n, u}(\omega_t^1, \omega_t^2)\,dt\right)\,dW_{x, x}^T(\omega^1)\,dW_{y, y}^T(\omega^2)\,dx\,dy\right]\|\psi_{n, u}\|^2,
\label{equation.lower.bound.second.localization}
\end{align}
where $\mathcal{A}_{i, j}^T$ is defined as
\begin{gather}
\alpha\int_0^T\!\!\!\int_0^t\!\!\int_{|k| \leq \Lambda}\frac{e^{-|k|(t - s)}}{|k|}e^{-ik(\omega_t^i - \omega_s^j)}\,dk\,ds\,dt.
\label{equation.action.piezo.electric.polaron}
\end{gather}
(See Subsection \ref{subsection.remark.functional.integrals} of the introduction.) Note how $\exp\left(-\int_0^T V_{n, u}(\omega_t)\,dt\right)$ may be conveniently expressed as $\Omega_{n, u}^T(\omega)$, the function equal to 1 if the 6-dimensional Brownian path $\omega_t$ is completely contained in $\Upsilon_{n, u}$ for all times $0 \leq t \leq T$, and zero otherwise. Note also how \eqref{equation.action.piezo.electric.polaron} can be computed explicitly, yielding
\begin{align}
4\pi\alpha\int_0^T\!\!\!\int_0^t &\frac{1}{|\omega_t^i - \omega_s^j|^2 + (t - s)^2}\nonumber\\
&\times\left\{1 - e^{-\Lambda(t - s)}\left[\frac{t - s}{|\omega_t^i - \omega_s^j|}\sin\left(\Lambda|\omega_t^i - \omega_s^j|\right) + \cos\left(\Lambda|\omega_t^i - \omega_s^j|\right)\right]\right\}\,ds\,dt.
\end{align}
It is easy to verify as well that the expression in braces is bounded below and above by 0 and 2, respectively. If now $\omega$ is a Brownian path such that $\Omega_{n, u}^T(\omega) = 1$,
\begin{align}
|\omega_t^1 - \omega_s^2| \geq & \, |\omega_t^1 - \omega_t^2| - |\omega_s^2 - \omega_t^2| \geq s_{n - 1} - 2R_n,\\
|\omega_t^2 - \omega_s^1| = |\omega_s^1 - \omega_t^2| \geq & \, |\omega_s^1 - \omega_s^2| - |\omega_t^2 - \omega_s^2| \geq s_{n - 1} - 2R_n,
\end{align}
and we then have that
\begin{align}
& \int\!\!\!\int\exp\left(\sum_{i, j = 2}^2\mathcal{A}_{i, j}^T\right)\Omega_{n, u}^T(\omega)\,dW_{x, x}^T(\omega_1)\,dW_{y, y}^T(\omega_2)\nonumber\\
\leq & \, \int\!\!\!\int\exp\left(\mathcal{A}_{1, 1}^T + \mathcal{A}_{2, 2}^T\right)\exp\left(\int_0^T\!\!\!\int_0^t\frac{16\pi\alpha}{(s_{n - 1} - 2R_n)^2 + (t - s)^2}\,ds\,dt\right)\,dW_{x, x}^T(\omega_1)\,dW_{y, y}^T(\omega_2),
\label{equation.functional.integral.estimate.piezo.electric.polaron}
\end{align}
and since
\begin{gather}
16\pi\alpha\int_0^T\!\!\!\int_0^t\frac{\,ds\,dt}{(s_{n - 1} - 2R_n)^2 + (t - s)^2} \leq 16\pi\alpha T\int_0^{\infty}\frac{dx}{(s_{n - 1} - 2R_n)^2 + x^2} = \frac{8\pi^2\alpha T}{s_{n - 1} - 2R_n},
\end{gather}
it follows that, by using the independence of $\mathcal{A}_{1, 1}^T$ and $\mathcal{A}_{2, 2}^T$, \eqref{equation.functional.integral.estimate.piezo.electric.polaron} is bounded above by
\begin{gather}
\exp\left(\frac{8\pi^2\alpha T}{s_{n - 1} - 2R_n}\right)\int\exp\left(\mathcal{A}_{1, 1}^T\right) dW_{x, x}^T\int\exp\left(\mathcal{A}_{2, 2}^T\right) dW_{y, y}^T.
\end{gather}
From this and the estimate \eqref{equation.lower.bound.second.localization}, we obtain that
\begin{align}
& \left(\psi_{n, u}, H_0\psi_{n, u}\right)\nonumber\\
\geq & \left\{-\lim_{R \to \infty}\lim_{T \to \infty}T^{-1}\log\left[\int_{B_R}\int\exp\left(\mathcal{A}_{1, 1}^T\right)dW_{x, x}^T\,dx\right]\right.\nonumber\\
& \,\,\,\, \left. - \lim_{R \to \infty}\lim_{T \to \infty}T^{-1}\log\left[\int_{B_R}\int\exp\left(\mathcal{A}_{2, 2}^T\right) dW_{y, y}^T\,dy\right] - \frac{8\pi^2\alpha}{s_{n - 1} - 2R_n}\right\}\|\psi_{n, u}\|^2\nonumber\\
= & \left(2E - \frac{8\pi^2\alpha}{s_{n - 1} - 2R_n}\right)\|\psi_{n, u}\|^2,
\label{equation.estimate.second.region}
\end{align}
where, as defined in the previous section, $E$ denotes the ground-state energy of the 1-electron piezoelectric polaron, $\text{inf spec }H^1$. (Right above Equation \eqref{equation.estimate.first.region}.) We then conclude, from Equation \eqref{equation.second.localization.performed},
\begin{gather}
(\psi_n, H_0\psi_n) \geq \left(2E - \frac{8\pi^2\alpha}{s_{n - 1} - 2R_n} - \frac{\pi^2}{2R_n^2}\right)\|\psi_n\|^2.
\end{gather}
By collecting terms, we conclude from this and the previous section that
\begin{align}
(\psi, H_A^2\psi) \geq & \left[2E - 8C_2(\Lambda)\alpha^2 - \frac{\pi^2}{2a_1^2} + \frac{A}{a_0 + a_1}\right]\|\psi_0\|^2\nonumber\\
& + \sum_{n = 1}^{\infty}\left[2E + \frac{A}{s_{n + 1}} - \frac{\pi^2}{2\min (a_{n + 1}, a_n)^2} - \frac{8\pi^2\alpha}{s_{n - 1} - 2R_n} - \frac{\pi^2}{2R_n^2}\right]\|\psi_n\|^2,
\end{align}
and so no-binding will occur if each one of the terms in brackets is greater than or equal to $2E$, or
\begin{gather}
A \geq \left[8C_2(\Lambda)(b_0 + b_1) + \frac{\pi^2(b_0 + b_1)}{2b_1^2}\right] \vee \bigvee_{n = 1}^{\infty}\left[\frac{\pi^2 t_{n + 1}}{2\min(b_{n + 1}, b_n)^2} + \frac{8\pi^2 t_{n + 1}}{t_{n - 1} - 2L_n} + \frac{\pi^2 t_{n + 1}}{2L_n^2}\right]\alpha,
\label{equation.condition.for.no.binding}
\end{gather}
where we made the substitutions $b_i = a_i\alpha$, $t_i = s_i\alpha$ and $L_n = R_n\alpha$ in order to factor out $\alpha$, and $\vee$ denotes maximum. The expression to the right in \eqref{equation.condition.for.no.binding} is certainly not $\infty$ if the parameters are chosen accordingly. For instance, by picking $L_n = t_{n - 1}/4$ and $b_i = bl^i$, for some $l > 1$, the right side of \eqref{equation.condition.for.no.binding} is less than or equal to
\begin{align}
& \,\left[8C_2(\Lambda)b(1 + l) + \frac{\pi^2(1 + l)}{2bl^2}\right]\vee\bigvee_{n = 1}^{\infty}\left[\frac{\pi^2 l^{n + 2}}{2b(l - 1)l^{2n}} + \frac{16\pi^2 l^{n + 2}}{l^n - 1} + \frac{8\pi^2 (l - 1)l^{n + 2}}{b(l^{n} - 1)^2}\right]\alpha\nonumber\\
\leq & \, \left[8C_2(\Lambda)b(1 + l) + \frac{\pi^2(1 + l)}{2bl^2}\right]\vee\left[\frac{\pi^2l}{2b(l - 1)} + \frac{16\pi^2l^3}{l - 1} + \frac{8\pi^2l^3}{b(l - 1)}\right]\alpha.
\end{align}
For a fixed value of $\Lambda$ one can in fact minimize the entire expression above, Equation \eqref{equation.condition.for.no.binding}, despite it being an infinite dimensional problem. This will be illustrated in the next section for the polaron model. In any case, the upshot is that there is an explicit function of $\Lambda$, which we shall call $C$, such that if $A \geq C(\Lambda)\alpha$, then no binding occurs.
\end{section}
\begin{section}{No-Binding in the Optical Polaron and Nelson Models}
\label{section.polaron.nelson}
The same calculation as above can be carried out for the optical polaron model, and one can minimize completely the final result, the analog of Equation \eqref{equation.condition.for.no.binding}, since there is no dependence on $\Lambda$. Computations here are entirely similar to those from the previous section, and very few changes have to be made. In the following, $H_A^2$ will be the two-electron optical polaron model Hamiltonian with Coulomb strength $A$ (as before for the piezoelectric polaron), $H^1$ will be the 1-electron analog, and $E$ will be the ground-state energy of $H^1$. We will make use of two inequalities, that we will now briefly explain how to obtain. The first inequality is $\text{inf spec }H_0^2 \geq -2\alpha - 2\alpha^2$ \cite{BT, B, B2, FLST}, which appears explicitly right below \cite[Equation (3.13)]{BT}, and can also be obtained by first replacing \cite[Equation (1.20)]{FLST} by the inequality $E \geq -\alpha - \alpha^2/4$, appearing below \cite[Equation (3.9)]{BT}, and then using \cite[Equation (2.26)]{FLST}. The second inequality is $E \leq -\alpha$ \cite{F, LLP, LP, G}, which follows from the arguments in Appendix B, but applied to the optical polaron. We then have that Equation \eqref{equation.estimate.first.region} changes to
\begin{gather}
(\psi_0, H_0^2\psi_0) \geq (-2\alpha - 2\alpha^2)\|\psi_0\|^2 \geq (2E - 2\alpha^2)\|\psi_0\|^2,
\end{gather}
whereas now \eqref{equation.estimate.second.region} becomes
\begin{gather}
\left(\psi_{n, u}, H_0\psi_{n, u}\right) \geq \left(2E - \frac{\sqrt{2}\alpha}{s_{n - 1} - 2R_n}\right)\|\psi_{n, u}\|^2,
\end{gather}
which follows immediately from the argument in the previous section leading to \eqref{equation.estimate.second.region}, but applied now to the bipolaron action
\begin{gather}
\sum_{i, j = 1}^2\int_0^T\!\!\!\int_0^t\frac{e^{-(t - s)}}{|X_t^i - X_s^j|}\,ds\,dt.
\end{gather}
From this, the no-binding condition (the analog of Equation \eqref{equation.condition.for.no.binding}) becomes
\begin{gather}
A \geq \left[2(b_0 + b_1) + \frac{\pi^2(b_0 + b_1)}{2b_1^2}\right]\vee\bigvee_{n = 1}^{\infty}\left[\frac{\pi^2 t_{n + 1}}{2\min(b_{n + 1}, b_n)^2} + \frac{\sqrt{2}t_{n + 1}}{t_{n - 1} - 2L_n} + \frac{\pi^2 t_{n + 1}}{2L_n^2}\right]\alpha.
\label{equation.no.binding.condition.bipolaron}
\end{gather}
Since the expression involves only numbers, it is amenable to minimization. Even though one can eliminate the variable $L_n$ by solving a cubic equation (which corresponds to optimizing the two last summands in the second bracket in \eqref{equation.no.binding.condition.bipolaron}), the computations involved in the elimination are so cumbersome that we will just content ourselves with simplifying the minimization problem by rescaling $x_n \equiv 2L_n/t_{n - 1}$, which significantly reduces the numerical work involved. We are led to minimizing $F : (0, \infty)^{\mathbb{N}}\times (0, 1)^{\mathbb{N}} \to (0, \infty)$, defined as
\begin{align}
& \, F((b_0, b_1, b_2, \ldots), (x_1, x_2, x_3, \ldots))\nonumber\\
= & \,\left[2(b_0 + b_1) + \frac{\pi^2(b_0 + b_1)}{2b_1^2}\right]\vee\bigvee_{n = 1}^{\infty}\left[\frac{\pi^2 t_{n + 1}}{2\min(b_{n + 1}, b_n)^2} + \frac{\sqrt{2}t_{n + 1}}{t_{n - 1}}\left(\frac{1}{1 - x_n} + \frac{\sqrt{2}\pi^2}{t_{n - 1}x_n^2}\right)\right]\nonumber\\
\equiv & \bigvee_{n = 0}^{\infty}F_n.
\label{equation.definition.F}
\end{align}
The minimization of $F$, in principle an infinite-dimensional problem, is much simpler than it seems, since it can be reduced to a low-dimensional one. Indeed, consider the truncated function $F_0\vee F_1 : (0, \infty)^3 \times (0, 1) \to (0, \infty)$ given by the maximum of the first two terms, that is
\begin{align}
F_0\vee F_1((b_0, b_1, b_2), x) = \left[2(b_0 + b_1) + \frac{\pi^2(b_0 + b_1)}{2b_1^2}\right]\vee\left[\frac{\pi^2 t_{2}}{2\min(b_{2}, b_1)^2} + \frac{\sqrt{2}t_{2}}{t_{0}}\left(\frac{1}{1 - x} + \frac{\sqrt{2}\pi^2}{t_{0}x^2}\right)\right].
\end{align}
$F_0\vee F_1$ is now so simple that it can be minimized directly through numerical optimization. We get that its minimum is smaller than 25.9 when the parameters $b_0 = 7.27, b_1 = b_2 = 3.44, x = 0.702$ are chosen. If we now pick $b_n = (n - 1)b_2 = (n - 1)\times 3.44$ for $n \geq 3$ and $x_n = x = 0.702$ for all $n \geq 1$, it is easy to show that $F_{n + 1} \leq F_n$ for all $n \geq 1$. In this way we have found that the minimum of the expression \eqref{equation.definition.F} is less than 25.9. In order to compare this with previous results in \cite{FLST, BB}, we ought to multiply by $\sqrt{2}$, as in those works the Laplacian was not divided by 2, and by doing so we get a number smaller than 36.7. This is a significant improvement over the previous results of 52.1 \cite{BB} and 53.2 \cite{FLST}: a reduction of more than 29.6\%.

We would like to remark that the partitions considered here have been indeed fully optimized: one may think, for instance, of taking not bumps in the first partition, but pills (meaning functions that are part sine, part straight line, part cosine); however, it is easy to see that one gets the optimal answer when all the straight-line segments are collapsed into a point, except for the first one. This is why we considered just one pill and let the other functions be bumps instead.

We close the main body of the present article by briefly commenting on what happens in the massless Nelson model case. A family of lower bounds was provided for that model in \cite{B}, but here we will pick just one of them. (One could certainly refine the following argument, but here we are mostly interested in illustrating a point, and not so much in sharpness.) For $N = 2$, by picking $\theta = 3/2$, we get, from (2.36) in \cite{B},
\begin{align}
E_{\alpha}^{\Lambda, 2} + Q_{\alpha}^{\Lambda, 2} \geq - D_1\alpha^2 - D_2\alpha^8,
\end{align}
where $D_1$ and $D_2$ are positive constants independent of $\Lambda$, $E_{\alpha}^{\Lambda, 2}$ is the ground-state energy of the 2-particle massless Nelson model with cutoff $\Lambda$, and $Q_{\alpha}^{\Lambda, 2}$ is a renormalizing term, defined earlier, in expression \eqref{equation.renormalizing.term}, with $N = 2$. (The term $\alpha^8$ is most likely spurious -- see \cite{B}.) In any case, since we know that $E_{\alpha}^{\Lambda, 1} + Q_{\alpha}^{\Lambda, 1} \leq 0$, where $E_{\alpha}^{\Lambda, 1}$ and $Q_{\alpha}^{\Lambda, 1}$ are the 1-particle analogs of $E_{\alpha}^{\Lambda, 2}$ and $Q_{\alpha}^{\Lambda, 2}$ (see Appendix B), all the steps to conclude no-binding are identical to those for the piezoelectric polaron, except that now the condition reads as
\begin{gather}
A \geq \left[D_1(b_0 + b_1) + D_2(b_0 + b_1)\alpha^6 + \frac{\pi^2(b_0 + b_1)}{2b_1^2}\right]\vee\bigvee_{n = 1}^{\infty}\left[\frac{\pi^2 t_{n + 1}}{2\min(b_{n + 1}, b_n)^2} + \frac{8\pi^2 t_{n + 1}}{t_{n - 1} - 2L_n} + \frac{\pi^2 t_{n + 1}}{2L_n^2}\right]\alpha,
\label{equation.condition.no.binding.nelson}
\end{gather}
and so, as can be seen, even though a no-binding condition is obtained, it is not linear in $\alpha$. From discussions appearing in \cite{B}, $D_2$ can probably be set equal to zero. If that is the case, then indeed one would get a condition identical to that for the piezoelectric polaron (even better, since it would not diverge with $\Lambda$). One may as well say that, after all, $\alpha$ ought to be in a certain range $[0, \beta]$, and by bounding $\alpha \leq \beta$, one could eliminate the higher order term in $\alpha$, and this would lead to a linear relationship $A \geq C\alpha$, as expected.

As said earlier, this is just an illustration, and many other no-binding conditions may be obtained from the general bound \cite[Equation (2.36)]{B}. As a final comment, when we say that no-binding holds for both the renormalized and unrenormalized theories, we mean this: even though in principle we have proven no-binding only for the unrenormalized Nelson Hamiltonian, the corresponding result for the renormalized theory follows immediately, as $E_{\alpha}^{\Lambda, 2} = 2E_{\alpha}^{\Lambda, 1}$ implies
\begin{gather}
E_{\alpha}^{\Lambda, 2} + Q_{\alpha}^{\Lambda, 2} = 2\left(E_{\alpha}^{\Lambda, 1} + Q_{\alpha}^{\Lambda, 1}\right),
\label{equation.energy.renormalization.nelson.two.particles}
\end{gather}
and the no-binding result in this case is obtained by taking $\Lambda \to \infty$. (See the discussion preceding \cite[Equation (1.2)]{GM} for technical remarks concerning this last limiting step.)
\end{section}
\begin{appendix}
\setcounter{section}{1}
\setcounter{equation}{0}
\begin{section}*{Appendix A: The Clark-Ocone Formula and its Use in the Obtention of Lower Bounds in Non-relativistic QFT}
Throughout this work, we made use of lower bounds for the spectrum of each one of the models involved. These bounds were made possible thanks to a technique that was presented in an article written by the author and L.E. Thomas \cite{BT}, which we would like to briefly sketch here. The main idea goes along the following lines. First, we fix a model from non-relativistic quantum field theory (for example, any of the three mentioned in this paper). The number of particles involved or the addition of repulsion between them are not important factors here; the argument is fairly general. We let $H$ be the Hamiltonain of the model chosen, and $N$ be the number of particles. If one can integrate the quantum field variables, which is typically the case, one will be able to find a Feynman-Kac-like formula
\begin{gather}
\inf\text{spec } H \geq -\limsup_{T \to \infty}T^{-1}\log\left\{\sup_{x \in \mathbb{R}^{3N}}E^x\left[\exp(\mathcal{A}_T)\right]\right\},
\end{gather}
for some functional $\mathcal{A}_T$ of $3N$-dimensional Brownian paths on $[0, T]$. ($E^x$ means expectation with respect to Brownian motion starting at $x$.) The Clark-Ocone formula then states that
\begin{gather}
\mathcal{A}_T = E^x(\mathcal{A}_T) + \int_0^T\rho_t\, dX_t,
\end{gather}
for some $\mathbb{R}^{3N}$-valued stochastic process $\rho$, and the integral appearing is It\^{o}. A supermartingale estimate allows one then to get the following bound:
\begin{gather}
E^x\left[\exp(\mathcal{A}_T)\right] \leq \exp\left[E^x\left(\mathcal{A}_T\right)\right]E^x\left[\exp\left(\frac{p^2}{2(p - 1)}\int_0^T\rho_t^2\,dt\right)\right]^{1 - 1/p}
\label{equation.supermartingale.estimate}
\end{gather}
for all $p > 1$. Upper bounds on both terms inside exponentials on the right side of \eqref{equation.supermartingale.estimate}, with the right growth rate in $T$ (linear or sublinear), and uniform in $x$, allow one then to obtain a lower bound on the spectrum of $H$. This is the essence of the method.

We have intentionally skipped many technicalities in the description just given. For the full details, the reader is referred to references \cite{BT} (the paper of the author and Thomas), \cite{B} (a lower bound on the Nelson model), and \cite{B2} (the Ph.D. thesis of the author, where, in particular, additional discussions appear).
\end{section}
\begin{section}*{Appendix B: Proof of Upper and Lower Bounds for the Piezoelectric Polaron}
\label{appendix.upper.lower.bounds.piezoelectric.polaron}
In the present appendix we will prove the upper and lower bounds stated above, in Section \ref{section.partition}, for the ground-state energy of the $2$-electron piezoelectric polaron. We will actually do it for any number of particles $N$. It relies heavily on a recent paper of the author \cite{B}, in which lower bounds for the renormalized Nelson model were found, using functional-integral methods developed by the author and L.E. Thomas in \cite{BT}. We start from an expression, \cite[Equation (2.18)]{B}, whose time-integral from 0 to $T$ leads eventually to a lower bound on the ground-state energy of the Nelson model, when interpreted as an action. It is given by
\begin{gather}
256\pi^2\alpha^2\sum_{m = 1}^N\left(\sum_{n = 1}^N\mathcal{C}_{m, n}\right)^2,
\end{gather}
for some positive $\mathcal{C}_{m, n}$'s that can be bounded from above as
\begin{align}
\mathcal{C}_{m, n} \leq \, & \int_0^{\Lambda}\!\!\!\int_0^u\frac{1 - e^{-(r + r^2/2)(T - u)}}{1 + r/2}re^{-r(u - s)}|\varphi(r|X_u^m - X_s^n + x^m - x^n|)|\,ds\,dr\nonumber\\
& \, + \int_0^{\Lambda}\!\!\!\int_u^T\!\!\int_u^t e^{-r(t - s)}e^{-r^2(t - u)/2}e^{-r^2(s - u)/2}|\varphi(r|X_u^m - X_u^n + x^m - x^n|)|r^2\,ds\,dt\,dr\nonumber\\
\equiv \, & \mathcal{D}_{m, n} + \mathcal{E}_{m, n},
\end{align}
where the vectors $X^n$ are independent 3D Brownian motions, and $\varphi(x)$ is the function $(\sin x - x\cos x)/x^2$. It was shown in \cite[Lemma 2.2]{B} that $\mathcal{E}_{m, n}$ is bounded above by $2^{-1}(1 - \delta_{nm})\|\varphi(x)/x\|_1$. A crude estimate on $\mathcal{D}$ allows us to bound it from above as
\begin{align}
& \int_0^{\Lambda}\!\!\!\int_0^u\frac{re^{-r(u - s)}}{1 + r/2}|\varphi(r|X_u^m - X_s^n + x^m - x^n|)|\,ds\,dr \leq \|\varphi\|_{\infty}\int_0^{\Lambda}\!\!\!\int_0^u\frac{re^{-r(u - s)}}{1 + r/2}\,ds\,dr\nonumber\\
\leq & \, \|\varphi\|_{\infty}\int_0^{\Lambda}\frac{dr}{1 + r/2} = 2\|\varphi\|_{\infty}\log(1 + \Lambda/2).
\end{align}
$\mathcal{C}_{m, n}$ is then bounded above by a logarithmically diverging function of $\Lambda$. It follows then from the analysis in \cite[Section 2]{B} that the ground-state energy of the piezo-electric polaron is bounded below by
\begin{align}
& -8\pi\alpha N\log(1 + \Lambda/2) - 32\pi^2\alpha^2 N^3\left[\|\varphi(x)/x\|_1 + 4\|\varphi\|_{\infty}\log(1 + \Lambda/2)\right]^2\nonumber\\
\equiv & -C_1(\Lambda)\alpha N - C_2(\Lambda)\alpha^2 N^3.
\label{definition.C1}
\end{align}

An upper bound for the 1-particle piezoelectric polaron, good enough for our purposes here, follows immediately from certain simplifications. First, we note that in the 1-particle case the starting and ending position of the Brownian path $\omega$ does not appear in the action
\begin{gather}
\alpha\int\!\!\!\int_0^T\!\!\!\int_0^t\chi_{\Lambda}(k)e^{-|k|(t - s)}e^{-ik(\omega_t - \omega_s)}|k|^{-1}\,ds\,dt\,dk
\end{gather}
(see Equation \eqref{equation.feynman.kac.exact.piezoelectric.polaron}), due to a cancellation in the difference of the Brownian path evaluated at different times. This fact allows us to write the ground-state energy of the 1-particle piezo-electric polaron as
\begin{gather}
-\lim_{T \to \infty}T^{-1}\log\left[\int\exp\left(\alpha\int\!\!\!\int_0^T\!\!\!\int_0^t\chi_{\Lambda}(k)e^{-|k|(t - s)}e^{-ik(\omega_t - \omega_s)}|k|^{-1}\,ds\,dt\,dk\right)dW_{0, 0}^T(\omega)\right].
\label{equation.piezoelectric.polaron.one.electron}
\end{gather}
Equation \eqref{equation.piezoelectric.polaron.one.electron} can be derived by adapting the analysis leading to \cite[Equation (5.3.48)]{R} to the piezoelectric polaron. Furthermore, by following the arguments in \cite[Section 3]{DV}, and making the necessary changes for the piezoelectric polaron, Brownian motion tied to 0 at both ends may be replaced by standard Brownian motion starting at 0, since when $T$ is large the terminal condition for a Brownian path $\omega$ is essentially irrelevant. With these simplifications, we obtain that the ground-state energy of the piezoelectric polaron can be written as
\begin{gather}
-\lim_{T \to \infty}T^{-1}\log\left\{E\left[\exp\left(\alpha\int\!\!\!\int_0^T\!\!\!\int_0^t\chi_{\Lambda}(k)e^{-|k|(t - s)}e^{-ik(X_t - X_s)}|k|^{-1}\,ds\,dt\,dk\right)\right]\right\}.
\label{equation.piezoelectric.polaron.one.electron.simplified}
\end{gather}
Now, from Jensen's inequality, we get that
\begin{align}
& E\left[\exp\left(\alpha\int\!\!\!\int_0^T\!\!\!\int_0^t\chi_{\Lambda}(k)e^{-|k|(t - s)}e^{-ik(X_t - X_s)}|k|^{-1}\,ds\,dt\,dk\right)\right]\nonumber\\
\geq & \exp\left(\alpha\int\!\!\!\int_0^T\!\!\!\int_0^t\chi_{\Lambda}(k)e^{-|k|(t - s)}|k|^{-1}E\left(e^{-ik(X_t - X_s)}\right)\,ds\,dt\,dk\right)\nonumber\\
= & \exp\left(\alpha\int\!\!\!\int_0^T\!\!\!\int_0^t\chi_{\Lambda}(k)e^{-|k|(t - s)}|k|^{-1}e^{-k^2(t - s)/2}ds\,dt\,dk\right)\nonumber\\
= & \exp\left[C_1(\Lambda)\alpha T + o(T)\right],
\end{align}
from which the upper bound $-C_1(\Lambda)\alpha$ for the piezoelectric polaron is obtained. (The last equality follows from Equation \eqref{definition.C1}.) Both inequalities, the lower bound \eqref{definition.C1} for $N = 1$ and the upper bound $-C_1(\Lambda)\alpha$ we just derived, agree for small $\alpha$ at the result one gets from second-order perturbation theory \cite{WGT}. Our lower bound \eqref{definition.C1} for $N = 1$ has the advantage over a previous lower bound for the piezoelectric polaron in \cite{TW} that an explicit answer is obtained, valid for all values of $\alpha$. In \cite{TW} the lower bound involves quantities that simplify only in limiting regimes of $\alpha$.

We would like to finish this appendix by pointing out a curious fact: One of the lower bounds for the piezoelectric polaron in \cite{TW} contains a term proportional to $\alpha\log\alpha$ that is not divergent in $\Lambda$, which is obtained under certain assumptions on $\alpha$. A term just like that was obtained in a lower bound in \cite{B} under a certain regime of $\alpha$, but squared. This seems to us more like a coincidence than an actual connection, since, as explained in \cite{B}, our logarithmically divergent term in $\alpha$ is probably not there really.
\end{section}
\setcounter{section}{2}
\setcounter{equation}{0}
\begin{section}*{\noindent Appendix C: Note on the Retarded Nature of the Polaron and Nelson Actions}
\label{appendix.second}
We would like to close the present article by addressing two mistakes made in Chapter 5 of the Ph.D. thesis of the author \cite{B2}. As the present article is based partly on that chapter from the thesis, we think it is relevant to resolve the problems here. First, in \cite[Section 5.1]{B2} we localized the six-dimensional vector $(x, y)$, representing the positions of two electrons in 3-space, in the region $\Omega \equiv \left\{(x, y) : |x - y| \geq d\right\}$. This localization should have the effect of keeping the two electrons far from each other; however, this is not what actually happens, at least not if the electrons interact with each other through their coupling to a polar crystal. The point is that the electrons do not attract each other instantaneously; rather, they attract their entire past histories in space, with different weights, with the remote past being less relevant than the recent times. The problem then is that merely localizing $(x, y)$ in $\Omega$ will mean that a Brownian path $(X_t, Y_t)$, representing a potential trajectory of the particles, satisfies $|X_t - Y_t| \geq d$; but it will not necessarily mean that $|X_t - Y_s| \geq d$ for all $t$ and $s$. The mistake made in the thesis was to assume that $|X_t - Y_s| \geq d$ held for all $t$ and $s$ just by localizing in $\Omega$. One further localization was missing, and is the one we added in this article, where we pinned one of the electrons, in which case the entire trajectories of the particles are now well-separated.

The other mistake was to neglect the factor $1/2$ in front of the localization error, arising in the IMS formula in \cite[Equation (5.5)]{B2}, due to the use of the operator $p^2/2$, instead of $p^2$. The errors combined yield a no-binding condition for the optical bipolaron given by $A \geq 17.8\alpha$, which, even though might well be true, was derived using an argument that was not completely correct, as we have just seen.
\end{section}
\end{appendix}

\end{document}